\documentclass[conference]{IEEEtran}
\IEEEoverridecommandlockouts
\usepackage{cite}
\usepackage{amsmath,amssymb,amsfonts}
\usepackage{algorithmic}
\usepackage{graphicx}
\usepackage{textcomp}
\usepackage{xcolor}
\def\BibTeX{{\rm B\kern-.05em{\sc i\kern-.025em b}\kern-.08em
    T\kern-.1667em\lower.7ex\hbox{E}\kern-.125emX}}

\usepackage{microtype}

\usepackage{amsthm,mathtools}
\usepackage{mathrsfs}
\usepackage{enumitem}
\usepackage[dvipsnames]{xcolor}
\usepackage[nameinlink,noabbrev]{cleveref}

\newtheorem{theorem}{Theorem}[section]
\newtheorem{proposition}[theorem]{Proposition}
\newtheorem{lemma}[theorem]{Lemma}
\newtheorem{corollary}[theorem]{Corollary}
\theoremstyle{definition}
\newtheorem{definition}[theorem]{Definition}
\newtheorem{example}[theorem]{Example}
\theoremstyle{remark}

\crefname{theorem}{theorem}{theorems}
\Crefname{theorem}{Theorem}{Theorems}
\crefname{proposition}{proposition}{propositions}
\Crefname{proposition}{Proposition}{Propositions}
\crefname{lemma}{lemma}{lemmas}
\Crefname{lemma}{Lemma}{Lemmas}
\crefname{corollary}{corollary}{corollaries}
\Crefname{corollary}{Corollary}{Corollaries}
\crefname{definition}{definition}{definitions}
\Crefname{definition}{Definition}{Definitions}
\crefname{example}{example}{examples}
\Crefname{example}{Example}{Examples}
\crefname{remark}{remark}{remarks}
\Crefname{remark}{Remark}{Remarks}

\DeclareMathOperator{\Tr}{Tr}
\DeclareMathOperator{\Diag}{Diag}
\DeclareMathOperator{\id}{id}

\newcommand{\CP}{\mathrm{CP}}
\newcommand{\DNN}{\mathrm{DNN}}
\newcommand{\R}{\mathbb{R}}
\newcommand{\N}{\mathbb{N}}

\allowdisplaybreaks

\begin{document}

\title{Exactness of the doubly nonnegative relaxation\\for qubit-output quantum channels
\thanks{This work is in part supported by the National Research Foundation of Korea (NRF, RS-2024-00451435 (20\%), RS-2024-00413957 (20\%)), Institute of Information \& communications Technology Planning \& Evaluation (IITP, RS-2025-02305453 (15\%), RS-2025-02273157 (15\%), RS-2025-25442149 (15\%), RS-2021-II211343 (15\%)) grant funded by the Ministry of Science and ICT (MSIT), Institute of New Media and Communications (INMAC), and the BK21 FOUR program of the Education, Artificial Intelligence Graduate School Program (Seoul National University), and Research Program for Future ICT Pioneers, Seoul National University in 2026.}
}

\author{\IEEEauthorblockN{Hyunho Cha}
\IEEEauthorblockA{\textit{Dept. of Electrical and Computer Engineering} \\
\textit{Seoul National University}\\
Seoul, Republic of Korea \\
aiden132435@cml.snu.ac.kr}
\and
\IEEEauthorblockN{Jungwoo Lee}
\IEEEauthorblockA{\textit{Dept. of Electrical and Computer Engineering} \\
\textit{Seoul National University}\\
Seoul, Republic of Korea \\
junglee@snu.ac.kr}
}

\maketitle

\begin{abstract}
The resource theory for nonnegativity of quantum amplitudes distinguishes completely positive completely positive (CPCP) quantum channels from the larger class of completely positive doubly nonnegative (CPDNN) quantum channels. Johnston and Sikora showed that all qubit-to-qubit quantum channels that are CPDNN are also CPCP. However, they left open the question of whether a qutrit-to-qubit quantum channel exists that is CPDNN but not CPCP. We prove that no such channel exists and, more generally, that every CPDNN quantum channel with qubit output is CPCP. Thus the doubly nonnegative relaxation is exact for all qubit-output quantum channels. Our argument yields an explicit structural picture in which, after a canonical permutation of the Choi matrix, every qubit-output CPDNN channel is determined by a nonnegative vector and a nonnegative matrix, subject to a single positive-semidefinite block constraint. This gives a complete binary-output normal form, an explicit formula for the action of the channel, and quantitative population-coherence tradeoff inequalities governing the unique output off-diagonal mode. In this sense, qubit output is the nontrivial binary regime in which trace preservation and double nonnegativity collapse the free-channel cone to a fully describable geometry.
\end{abstract}


\section{Introduction}

The resource theory for nonnegativity of quantum amplitudes introduced by Johnston and Sikora~\cite{johnston2022completely} studies quantum states whose density matrices lie in the cone of completely positive matrices and the quantum channels that preserve this cone. This framework is motivated by the role of nonnegative-amplitude states in connection with stoquastic Hamiltonians and the sign problem, and it provides a natural language for asking when nonnegativity is preserved under physical dynamics \cite{troyer2005computational, bravyi2006complexity, bravyi2010complexity}. In this setting, the free channels are the completely positive completely positive maps (CPCP), namely the maps that preserve complete positivity of matrices even after tensoring with an arbitrary ancilla.

A natural relaxation is obtained by replacing the completely positive cone by the larger cone of doubly nonnegative matrices \cite{hall1963copositive, berman2003completely, johnston2022completely}. The resulting maps are called completely positive doubly nonnegative (CPDNN). The Choi-matrix characterizations of~\cite{johnston2022completely} show that a linear map is CPCP precisely when its Choi matrix is completely positive, whereas it is CPDNN precisely when its Choi matrix is doubly nonnegative. Since doubly nonnegative matrices are significantly easier to recognize than completely positive matrices, it is natural to ask when the relaxation from CPCP to CPDNN is exact \cite{dickinson2014computational, dur2021conic, johnston2022completely}.

For arbitrary linear maps, the two notions separate beyond the low-dimensional cases considered in~\cite{johnston2022completely}. For quantum channels, however, trace preservation imposes additional constraints on the Choi matrix. Johnston and Sikora therefore asked whether there exists a qutrit-to-qubit quantum channel $\Phi : M_3 \to M_2$ that is CPDNN but not CPCP~\cite[Section~5.2]{johnston2022completely}. This is the smallest nontrivial qubit-output case in which complete positivity of the Choi matrix is no longer automatic from the ambient matrix size alone \cite{diananda1962non, maxfield1962matrix, gray1980nonnegative, burer2009difference}.

The principal result of this paper is that the open qutrit-to-qubit separation does not exist. In fact, we prove the stronger statement that every CPDNN quantum channel $\Phi : M_n \to M_2$ is CPCP for every input dimension $n$. Equivalently, the doubly nonnegative relaxation is exact for all qubit-output quantum channels.

The proof reveals a broader structural phenomenon that is specific to binary output. After a canonical permutation of coordinates, every qubit-output CPDNN channel has a Choi matrix determined by a vector $p$ of diagonal parameters and an entrywise nonnegative matrix $B$, subject to a single positive-semidefinite block constraint. This is a full channel-level normal form, not merely a membership test for a cone. It shows that the output diagonal is controlled by the vector $p$, while the unique output off-diagonal degree of freedom is governed by the single matrix $B$. In the fixed standard basis, we refer to diagonal entries as \emph{populations} and off-diagonal entries as \emph{coherences}. In this language, a qubit output has exactly one independent output coherence mode.

The same normal form yields quantitative inequalities that constrain how strongly the input entries may contribute to that unique output coherence. These inequalities exhibit a population-coherence tradeoff: the coefficients in $B$ are controlled by the diagonal splitting parameters $p_i$ and $1-p_j$, and the row and column bounds show that they cannot all be simultaneously large. Thus the structural consequences obtained here are intrinsic to the $2\times 2$ output block algebra itself.

\section{Preliminaries}\label{sec:prelim}

Throughout, $M_n$ denotes the algebra of $n\times n$ complex matrices, $I_n$ denotes the $n\times n$ identity matrix, and $E_{ij}\in M_n$ denotes the matrix unit with a $1$ in position $(i,j)$ and $0$ elsewhere. A quantum channel is a linear map $\Phi : M_n \to M_m$ that is completely positive as a linear map and trace-preserving \cite{stinespring1955positive, kraus1971general, watrous2018theory}. Its Choi matrix is
$$
J(\Phi):=\sum_{i,j=1}^n E_{ij}\otimes \Phi(E_{ij})\in M_n\otimes M_m\cong M_{nm}
$$
\cite{jamiolkowski1972linear, choi1975completely, jiang2013channel}.
All entrywise notions are understood in the fixed standard basis.

\begin{definition}
Let $r\in\N$.
\begin{enumerate}[label=\textup{(\arabic*)},leftmargin=2.2em]
\item A symmetric matrix $A\in \R^{r\times r}$ is \emph{completely positive} if there exist $s\in\N$ and vectors $x_1,\dots,x_s\in\R_{\ge 0}^r$ such that
$
A=\sum_{t=1}^s x_t x_t^\top
$
\cite{markham1971factorizations, berman1988complete, dickinson2011geometry}.
We write $A\in \CP_r$.
\item A symmetric matrix $A\in \R^{r\times r}$ is \emph{doubly nonnegative} if $A\succeq 0$ and every entry of $A$ is nonnegative. We write $A\in \DNN_r$.
\end{enumerate}
\end{definition}

Clearly $\CP_r\subseteq \DNN_r$ for every $r\in\N$.

\begin{definition}
A linear map $\Phi : M_n \to M_m$ is called \emph{completely positive completely positive}, abbreviated \emph{CPCP}, if for every $k\in\N$ and every $X\in \CP_{kn}$ one has
\[
(\id_k\otimes \Phi)(X)\in \CP_{km}.
\]
It is called \emph{completely positive doubly nonnegative}, abbreviated \emph{CPDNN}, if for every $k\in\N$ and every $X\in \DNN_{kn}$ one has
\[
(\id_k\otimes \Phi)(X)\in \DNN_{km}.
\]
\end{definition}

We also use the graph $G(A)$ of a symmetric matrix $A=(a_{ij})\in\R^{r\times r}$: this is the simple graph on $\{1,\dots,r\}$ with an edge $\{i,j\}$ whenever $i\neq j$ and $a_{ij}\neq 0$ \cite{berman1987combinatorial, kogan1993characterization}.

The following three facts will be used repeatedly.
\begin{itemize}[leftmargin=1.8em]
\item By \cite[Theorems~1 and~4]{johnston2022completely}, a linear map $\Phi : M_n\to M_m$ is CPCP if and only if $J(\Phi)\in\CP_{nm}$, and it is CPDNN if and only if $J(\Phi)\in\DNN_{nm}$.
\item By the standard Choi correspondence, a matrix $J\in M_n\otimes M_m$ is the Choi matrix of a completely positive map if and only if $J\succeq 0$ \cite{choi1975completely, wilde2013quantum}.
\item By \cite[Theorem~3.1]{berman1988bipartite}, if $A\in\DNN_r$ and $G(A)$ is bipartite, then $A\in\CP_r$.
\end{itemize}

The next two lemmas provide simple consequences that will be needed later.

\begin{lemma}\label{lem:traceblocks}
Let $\Phi : M_n \to M_m$ be trace-preserving, and write
$
J(\Phi)=\sum_{i,j=1}^n E_{ij}\otimes J_{ij},
$
$
J_{ij}\in M_m.
$
Then
$
\Tr(J_{ij})=\delta_{ij}
$
for all $1\le i,j\le n$.
\end{lemma}

\begin{proof}
For each pair $(i,j)$ we have $J_{ij}=\Phi(E_{ij})$ by definition of the Choi matrix. Since $\Phi$ is trace-preserving,
$
\Tr(J_{ij})=\Tr(\Phi(E_{ij}))=\Tr(E_{ij})=\delta_{ij}.
$
\end{proof}

\begin{lemma}\label{lem:permcp}
Let $A\in \CP_r$, and let $P$ be an $r\times r$ permutation matrix. Then $PAP^\top\in \CP_r$.
\end{lemma}

\begin{proof}
Because $A\in\CP_r$, there exist $s\in\N$ and vectors $x_1,\dots,x_s\in\R^r_{\ge 0}$ such that
$
A=\sum_{t=1}^s x_t x_t^\top.
$
Multiplying on the left by $P$ and on the right by $P^\top$ gives
$
PAP^\top=\sum_{t=1}^s (Px_t)(Px_t)^\top.
$
A permutation matrix merely reorders the coordinates of a vector, so each $Px_t$ is again entrywise nonnegative. Therefore $PAP^\top$ is a sum of outer products of entrywise nonnegative vectors, which means that $PAP^\top\in\CP_r$.
\end{proof}

\section{Binary-output normal form}\label{sec:normal}

We now isolate the rigid block structure forced by trace preservation and double nonnegativity in the qubit-output case.

\begin{lemma}\label{lem:collapse}
Let
$
J=\sum_{i,j=1}^n E_{ij}\otimes J_{ij}\in M_n\otimes M_2\cong M_{2n}
$
with $J\in\DNN_{2n}$ and $\Tr(J_{ij})=\delta_{ij}$ for all $1\le i,j\le n$. Write
\[
J_{ij}=\begin{pmatrix}
a_{ij} & b_{ij}\\
c_{ij} & d_{ij}
\end{pmatrix}
\qquad (1\le i,j\le n).
\]
Then for every $i\neq j$ one has
\[
a_{ij}=d_{ij}=0,
\qquad\text{and hence}\qquad
J_{ij}=\begin{pmatrix}
0 & b_{ij}\\
b_{ji} & 0
\end{pmatrix}.
\]
\end{lemma}

\begin{proof}
Because $J\in\DNN_{2n}$, the matrix $J$ is real symmetric, positive semidefinite, and entrywise nonnegative. In particular,
$
a_{ij},\ b_{ij},\ c_{ij},\ d_{ij}\ge 0
$
for all $i,j$, and symmetry of $J$ implies
$
a_{ij}=a_{ji},
$
$
d_{ij}=d_{ji},
$
$
c_{ij}=b_{ji}.
$
Now fix distinct indices $i\neq j$. The trace condition gives
$
a_{ij}+d_{ij}=\Tr(J_{ij})=\delta_{ij}=0.
$
Thus, $a_{ij}=0$ and $d_{ij}=0$.
\end{proof}

\begin{theorem}[Matrix-level binary-output normal form]\label{thm:matrixnormal}
Let
$
J=\sum_{i,j=1}^n E_{ij}\otimes J_{ij}\in M_n\otimes M_2\cong M_{2n}.
$
Then the following are equivalent.
\begin{enumerate}[label=\textup{(\roman*)},leftmargin=2.2em]
\item $J\in\DNN_{2n}$ and $\Tr(J_{ij})=\delta_{ij}$ for all $1\le i,j\le n$.
\item There exist a vector $p=(p_1,\dots,p_n)\in[0,1]^n$ and an entrywise nonnegative matrix $B=(b_{ij})\in\R_{\ge 0}^{n\times n}$ such that, if $P$ denotes the permutation matrix that reorders the standard basis of $\mathbb{C}^n\otimes\mathbb{C}^2$ from
\[
(1,1),(1,2),(2,1),(2,2),\dots,(n,1),(n,2)
\]
to
\[
(1,1),(2,1),\dots,(n,1),(1,2),(2,2),\dots,(n,2),
\]
then
\begin{align*}
PJP^\top & =
\begin{pmatrix}
D_p & B \\
B^\top & I_n-D_p
\end{pmatrix},\\
D_p & :=\Diag(p_1,\dots,p_n),
\end{align*}
and the displayed block matrix is positive semidefinite.
\end{enumerate}
\end{theorem}

\begin{proof}
Assume first that \textup{(i)} holds. Write
\[
J_{ij}=\begin{pmatrix}
a_{ij} & b_{ij}\\
c_{ij} & d_{ij}
\end{pmatrix}
\qquad (1\le i,j\le n).
\]
By Lemma~\ref{lem:collapse}, for every $i\neq j$ we have $a_{ij}=d_{ij}=0$ and $c_{ij}=b_{ji}$. For each diagonal block $J_{ii}$, the matrix $J$ is entrywise nonnegative, so $a_{ii},d_{ii}\ge 0$. Since $\Tr(J_{ii})=1$, we obtain
$
a_{ii}+d_{ii}=1.
$
Define
$
p_i:=a_{ii}\in[0,1]
$
and
$
B:=(b_{ij})_{i,j=1}^n.
$
Then $d_{ii}=1-p_i$ for every $i$, and the block matrices take the form
\[
J_{ii}=\begin{pmatrix}
p_i & b_{ii}\\
b_{ii} & 1-p_i
\end{pmatrix},
\qquad
J_{ij}=\begin{pmatrix}
0 & b_{ij}\\
b_{ji} & 0
\end{pmatrix}
\quad (i\neq j).
\]
Now apply the permutation $P$ described in the statement. Because $P$ merely groups all output-$0$ coordinates first and all output-$1$ coordinates second, the reordered matrix $A:=PJP^\top$ has the block form
\begin{equation}
\label{eq:A_definition}
A=
\begin{pmatrix}
D_p & B \\
B^\top & I_n-D_p
\end{pmatrix}.
\end{equation}
Since $P$ is orthogonal and $J\succeq 0$, we have $A\succeq 0$. Moreover, the entries of $B$ are entries of the entrywise nonnegative matrix $J$, so $B\ge 0$ entrywise. This proves \textup{(ii)}.

Conversely, assume that \textup{(ii)} holds, and set $A$ as in Eq.~\eqref{eq:A_definition}.
Because $p_i\in[0,1]$ for every $i$ and $B\ge 0$ entrywise, every entry of $A$ is nonnegative. The hypothesis also gives $A\succeq 0$. Therefore $A\in\DNN_{2n}$. Let
$
J:=P^\top A P.
$
As $P$ is a permutation matrix, $J$ is obtained from $A$ by permuting rows and columns, so $J$ is still positive semidefinite and entrywise nonnegative. Hence $J\in\DNN_{2n}$.

It remains to verify the block-trace conditions. Reading the block structure of $J$ from $A$, we have
\[
J_{ii}=\begin{pmatrix}
p_i & b_{ii}\\
b_{ii} & 1-p_i
\end{pmatrix},
\qquad
J_{ij}=\begin{pmatrix}
0 & b_{ij}\\
b_{ji} & 0
\end{pmatrix}
\quad (i\neq j).
\]
Therefore,
$
\Tr(J_{ii})=p_i+(1-p_i)=1
$
and
$
\Tr(J_{ij})=0 \quad (i\neq j).
$
Equivalently, $\Tr(J_{ij})=\delta_{ij}$ for all $i,j$. This proves \textup{(i)} and completes the equivalence.
\end{proof}

\begin{corollary}[Channel-level binary-output normal form]\label{cor:channelnormal}
Let $\Phi : M_n\to M_2$ be a CPDNN quantum channel. Then there exist $p=(p_1,\dots,p_n)\in[0,1]^n$ and an entrywise nonnegative matrix $B=(b_{ij})\in\R_{\ge 0}^{n\times n}$ such that
\[
P J(\Phi) P^\top=
\begin{pmatrix}
D_p & B \\
B^\top & I_n-D_p
\end{pmatrix}\succeq 0,
\]
where $P$ is the permutation from \Cref{thm:matrixnormal}. Conversely, every pair $(p,B)$ with these properties determines a CPDNN quantum channel with this Choi matrix.
\end{corollary}

\begin{proof}
If $\Phi$ is a CPDNN quantum channel, then the characterization of \cite{johnston2022completely} recalled in \Cref{sec:prelim} gives $J(\Phi)\in\DNN_{2n}$. Because $\Phi$ is trace-preserving, Lemma~\ref{lem:traceblocks} shows that the block traces satisfy $\Tr(J_{ij})=\delta_{ij}$. Applying \Cref{thm:matrixnormal} yields the required representation.

Conversely, let $(p,B)$ satisfy the stated conditions, and define
\[
A:=\begin{pmatrix}
D_p & B \\
B^\top & I_n-D_p
\end{pmatrix}\succeq 0,
\qquad
J:=P^\top A P.
\]
By \Cref{thm:matrixnormal}, the matrix $J$ belongs to $\DNN_{2n}$ and has block traces $\Tr(J_{ij})=\delta_{ij}$. By the standard Choi correspondence, $J$ is the Choi matrix of a completely positive map $\Phi : M_n\to M_2$, and the trace identities imply that $\Phi$ is trace-preserving. Hence $\Phi$ is a quantum channel. Finally, since $J\in\DNN_{2n}$, the characterization of \cite{johnston2022completely} recalled in \Cref{sec:prelim} shows that $\Phi$ is CPDNN.
\end{proof}

\begin{corollary}[Explicit action of the channel]\label{cor:action}
Let $\Phi : M_n\to M_2$ be a CPDNN quantum channel, and let $p$ and $B=(b_{ij})$ be as in Corollary~\ref{cor:channelnormal}. Then for every matrix $X=(x_{ij})\in M_n$,
\[
\Phi(X)=
\begin{pmatrix}
\sum_{i=1}^n p_i x_{ii} & \sum_{i,j=1}^n b_{ij}x_{ij} \\
\sum_{i,j=1}^n b_{ji}x_{ij} & \sum_{i=1}^n (1-p_i)x_{ii}
\end{pmatrix}.
\]
In particular, if $X$ is Hermitian, then the lower-left entry is the complex conjugate of the upper-right entry.
\end{corollary}

\begin{proof}
Using the block decomposition of the Choi matrix obtained in the proof of \Cref{thm:matrixnormal}, we have
\[
\Phi(E_{ii})=J_{ii}=\begin{pmatrix}
p_i & b_{ii}\\
b_{ii} & 1-p_i
\end{pmatrix}
\]
and
\[
\Phi(E_{ij})=J_{ij}=\begin{pmatrix}
0 & b_{ij}\\
b_{ji} & 0
\end{pmatrix}
\quad (i\neq j).
\]
Now let $X=(x_{ij})\in M_n$. Since $X=\sum_{i,j=1}^n x_{ij}E_{ij}$ and $\Phi$ is linear,
$
\Phi(X)=\sum_{i,j=1}^n x_{ij}\Phi(E_{ij}).
$
Splitting the sum into diagonal and off-diagonal terms gives
\begin{align*}
\Phi(X)
&=\sum_{i=1}^n x_{ii}
\begin{pmatrix}
p_i & b_{ii}\\
b_{ii} & 1-p_i
\end{pmatrix}
+\sum_{\substack{i,j=1\\ i\neq j}}^n x_{ij}
\begin{pmatrix}
0 & b_{ij}\\
b_{ji} & 0
\end{pmatrix}\\
&=
\begin{pmatrix}
\sum_{i=1}^n p_i x_{ii} & \sum_{i,j=1}^n b_{ij}x_{ij} \\
\sum_{i,j=1}^n b_{ji}x_{ij} & \sum_{i=1}^n (1-p_i)x_{ii}
\end{pmatrix}.
\end{align*}
This proves the formula.

Assume now that $X$ is Hermitian. Then $x_{ji}=\overline{x_{ij}}$ for all $i,j$. Using the fact that the coefficients $b_{ij}$ are real, we obtain
\begin{align*}
\sum_{i,j=1}^n b_{ji}x_{ij}
&=\sum_{i,j=1}^n b_{ij}x_{ji}
=\sum_{i,j=1}^n b_{ij}\overline{x_{ij}}
=\overline{\sum_{i,j=1}^n b_{ij}x_{ij}}.
\end{align*}
Therefore the lower-left entry of $\Phi(X)$ is the complex conjugate of the upper-right entry.
\end{proof}

The preceding corollary makes the binary-output geometry transparent. The output diagonal depends only on the diagonal of the input, whereas the entire output off-diagonal part is encoded by the single scalar functional
$
L_B(X):=\sum_{i,j=1}^n b_{ij}x_{ij}.
$
This is the precise sense in which a qubit output has one independent coherence mode.

\section{Population-coherence tradeoff}\label{sec:tradeoff}

The normal form of Corollary~\ref{cor:channelnormal} yields quantitative restrictions on the coefficients that generate the unique output coherence mode.

\begin{proposition}[Weighted contraction factorization]\label{prop:contraction}
Let $p=(p_1,\dots,p_n)\in[0,1]^n$ and let $B=(b_{ij})\in\R_{\ge 0}^{n\times n}$ be such that
\[
A:=\begin{pmatrix}
D_p & B \\
B^\top & I_n-D_p
\end{pmatrix}\succeq 0.
\]
Define index sets
\[
I:=\{i: p_i>0\},
\qquad
J:=\{j: p_j<1\}.
\]
Then every row of $B$ whose index lies outside $I$ is zero, and every column of $B$ whose index lies outside $J$ is zero. If both $I$ and $J$ are nonempty, define
\begin{gather*}
D_I:=\Diag(p_i)_{i\in I},
\qquad
F_J:=\Diag(1-p_j)_{j\in J},\\
R:=D_I^{-1/2}B_{I,J}F_J^{-1/2}.
\end{gather*}
Then $R$ is entrywise nonnegative and satisfies
$
\|R\|_{2\to 2}\le 1.
$
Equivalently,
$
B_{I,J}=D_I^{1/2}RF_J^{1/2}
$
with $R\ge 0$ entrywise and $\|R\|_{2\to 2}\le 1$.
If either $I$ or $J$ is empty, then necessarily $B=0$.
\end{proposition}

\begin{proof}
We first show that rows indexed by $I^c$ vanish. Fix $i\notin I$, so $p_i=0$, and let $j\in\{1,\dots,n\}$ be arbitrary. Consider the $2\times 2$ principal submatrix of $A$ obtained by selecting the $i$th row and column from the first block and the $j$th row and column from the second block. This principal submatrix is
\[
\begin{pmatrix}
0 & b_{ij}\\
b_{ij} & 1-p_j
\end{pmatrix}.
\]
Because principal submatrices of positive-semidefinite matrices are positive semidefinite, the displayed matrix is positive semidefinite \cite{horn2012matrix}. Its determinant is $-b_{ij}^2$, so positive semidefiniteness forces $b_{ij}=0$. Since $j$ was arbitrary, the entire $i$th row of $B$ is zero.

The argument for columns indexed by $J^c$ is analogous. Fix $j\notin J$, so $p_j=1$. For arbitrary $i$, the relevant $2\times 2$ principal submatrix is
\[
\begin{pmatrix}
p_i & b_{ij}\\
b_{ij} & 0
\end{pmatrix},
\]
which is again positive semidefinite. Its determinant is $-b_{ij}^2$, hence $b_{ij}=0$. Thus the entire $j$th column of $B$ is zero.

If either $I$ or $J$ is empty, then every row or every column of $B$ vanishes, so $B=0$ and the remaining assertions are immediate. Assume from now on that both $I$ and $J$ are nonempty. After deleting the zero rows and zero columns identified above, the corresponding principal submatrix of $A$ is
\[
A_{I,J}:=\begin{pmatrix}
D_I & B_{I,J} \\
B_{I,J}^\top & F_J
\end{pmatrix}\succeq 0.
\]
Because $p_i>0$ for $i\in I$ and $1-p_j>0$ for $j\in J$, the diagonal matrices $D_I$ and $F_J$ are positive definite. Conjugating $A_{I,J}$ by the block-diagonal matrix
\[
\begin{pmatrix}
D_I^{-1/2} & 0 \\
0 & F_J^{-1/2}
\end{pmatrix}
\]
preserves positive semidefiniteness and yields
\[
\begin{pmatrix}
I_{|I|} & R \\
R^\top & I_{|J|}
\end{pmatrix}\succeq 0.
\]
Since the upper-left block is the identity matrix, its Schur complement is positive semidefinite \cite{haynsworth1968determination, zhang2006schur}. Therefore
$
I_{|J|}-R^\top R\succeq 0.
$
This implies that every singular value of $R$ is at most $1$, or equivalently $\|R\|_{2\to 2}\le 1$ \cite{bhatia2009positive, horn2012matrix}. The factorization
$
B_{I,J}=D_I^{1/2}RF_J^{1/2}
$
is simply the definition of $R$ rewritten. Finally, the matrices $D_I^{-1/2}$ and $F_J^{-1/2}$ are diagonal with strictly positive entries, so rescaling the rows and columns of the entrywise nonnegative matrix $B_{I,J}$ preserves entrywise nonnegativity. Hence $R\ge 0$ entrywise.
\end{proof}

\begin{corollary}[Population-coherence tradeoff]\label{cor:tradeoff}
Under the assumptions of Proposition~\ref{prop:contraction}, the coefficients of $B$ satisfy the pointwise bounds
\[
0\le b_{ij}\le \sqrt{p_i(1-p_j)}
\qquad\text{for all }1\le i,j\le n.
\]
Moreover, for every $i$ with $p_i>0$,
\[
\sum_{j\, :\, p_j<1}\frac{b_{ij}^2}{1-p_j}\le p_i,
\]
and for every $j$ with $p_j<1$,
\[
\sum_{i\, :\, p_i>0}\frac{b_{ij}^2}{p_i}\le 1-p_j.
\]
\end{corollary}

\begin{proof}
Fix indices $i$ and $j$. The $2\times 2$ principal submatrix of $A$ using the $i$th row and column from the first block and the $j$th row and column from the second block is
\[
\begin{pmatrix}
p_i & b_{ij}\\
b_{ij} & 1-p_j
\end{pmatrix}.
\]
Because $A\succeq 0$, this principal submatrix is positive semidefinite. In particular, its determinant is nonnegative:
$
p_i(1-p_j)-b_{ij}^2\ge 0.
$
Since also $b_{ij}\ge 0$, we obtain
$
0\le b_{ij}\le \sqrt{p_i(1-p_j)}.
$
This proves the pointwise bound.

Next assume $p_i>0$. If there is no $j$ with $p_j<1$, then the sum is empty and the inequality is trivial. Otherwise, the set
$
J=\{j:p_j<1\}
$
is nonempty and by Proposition~\ref{prop:contraction} we may form the matrix $R=D_I^{-1/2}B_{I,J}F_J^{-1/2}$ with $\|R\|_{2\to 2}\le 1$. The operator norm bound implies
$
RR^\top\preceq I_{|I|}.
$
Taking the diagonal entry corresponding to the row indexed by $i$ gives
\[
\sum_{j\, :\, p_j<1} \frac{b_{ij}^2}{p_i(1-p_j)}\le 1.
\]
Multiplying by $p_i$ yields
\[
\sum_{j\, :\, p_j<1} \frac{b_{ij}^2}{1-p_j}\le p_i.
\]
This is the row inequality.

The column inequality is similar. Fix $j$ with $p_j<1$. If there is no $i$ with $p_i>0$, the sum is empty and there is nothing to prove. Otherwise, $\|R\|_{2\to 2}\le 1$ implies
$
R^\top R\preceq I_{|J|}.
$
Taking the diagonal entry corresponding to the column indexed by $j$ gives
\[
\sum_{i\, :\, p_i>0} \frac{b_{ij}^2}{p_i(1-p_j)}\le 1.
\]
Multiplying by $1-p_j$ yields
\[
\sum_{i\, :\, p_i>0} \frac{b_{ij}^2}{p_i}\le 1-p_j.
\]
This proves the stated column bound.
\end{proof}

The interpretation is immediate from Corollary~\ref{cor:action}. The parameter $p_i$ describes how the $i$th input basis state is split between the two output populations, while the coefficient $b_{ij}$ measures how strongly the input entry $x_{ij}$ may contribute to the unique output coherence. Corollary~\ref{cor:tradeoff} therefore gives a genuine population-coherence tradeoff. In particular, the row and column inequalities show that each row and column of $B$ carries only a finite coherence budget determined by the diagonal parameters.

\begin{example}[The single-input case]\label{ex:n1}
If $n=1$, then the normal form reduces to a single output state
\[
\begin{pmatrix}
p & b\\
b & 1-p
\end{pmatrix}\succeq 0,
\qquad 0\le p\le 1,
\qquad b\ge 0.
\]
Positive semidefiniteness is equivalent to the scalar inequality
\[
0\le b\le \sqrt{p(1-p)}.
\]
Thus the maximal achievable coherence occurs at $p=\tfrac12$, whereas the coherence must be small when the population is concentrated near one basis state. Corollary~\ref{cor:tradeoff} is the many-input version of this same phenomenon.
\end{example}

\section{Exactness of the CPDNN relaxation}\label{sec:exactness}

We now return to the main exactness question.

\begin{theorem}[Matrix exactness for qubit output]\label{thm:matrixexact}
Let
$
J=\sum_{i,j=1}^n E_{ij}\otimes J_{ij}\in M_n\otimes M_2\cong M_{2n}.
$
Assume that $J\in\DNN_{2n}$ and that $\Tr(J_{ij})=\delta_{ij}$ for all $1\le i,j\le n$. Then $J\in\CP_{2n}$.
\end{theorem}

\begin{proof}
By \Cref{thm:matrixnormal}, there exist $p\in[0,1]^n$, an entrywise nonnegative matrix $B\in\R_{\ge 0}^{n\times n}$, and the canonical permutation matrix $P$ such that
\[
A:=PJP^\top=
\begin{pmatrix}
D_p & B \\
B^\top & I_n-D_p
\end{pmatrix}\succeq 0.
\]
Because $A$ is positive semidefinite and entrywise nonnegative, we have $A\in\DNN_{2n}$.

Consider the graph $G(A)$. The upper-left block $D_p$ and the lower-right block $I_n-D_p$ are diagonal matrices, so there are no edges of $G(A)$ joining two vertices both lying in $\{1,\dots,n\}$ and no edges joining two vertices both lying in $\{n+1,\dots,2n\}$. Any off-diagonal nonzero entry of $A$ must therefore occur in one of the off-diagonal blocks $B$ or $B^\top$, and hence every edge joins a vertex from the first set to a vertex from the second set. It follows that $G(A)$ is bipartite, with bipartition
\[
\{1,\dots,n\}\cup\{n+1,\dots,2n\}.
\]
Therefore, by \cite[Theorem~3.1]{berman1988bipartite}, $A\in\CP_{2n}$. Finally, Lemma~\ref{lem:permcp} shows that
$
J=P^\top A P\in\CP_{2n}.
$
\end{proof}

\begin{theorem}[Exactness for qubit-output quantum channels]\label{thm:channelexact}
For every $n\in\N$, every CPDNN quantum channel $\Phi : M_n\to M_2$ is CPCP.
\end{theorem}

\begin{proof}
Let $\Phi : M_n\to M_2$ be a CPDNN quantum channel. Because $\Phi$ is CPDNN, the characterization of \cite{johnston2022completely} recalled in \Cref{sec:prelim} shows that $J(\Phi)\in\DNN_{2n}$. Because $\Phi$ is trace-preserving, Lemma~\ref{lem:traceblocks} gives
$
\Tr(J_{ij})=\delta_{ij}
$
for all $1\le i,j\le n$, where $J(\Phi)=\sum_{i,j=1}^n E_{ij}\otimes J_{ij}$.
Thus $J(\Phi)$ satisfies both hypotheses of \Cref{thm:matrixexact}, and we conclude that $J(\Phi)\in\CP_{2n}$. This implies that $\Phi$ is CPCP.
\end{proof}

\begin{corollary}\label{cor:qutritqubit}
There does not exist a qutrit-to-qubit quantum channel $\Phi : M_3\to M_2$ that is CPDNN but not CPCP.
\end{corollary}

\begin{corollary}\label{cor:equality}
For every $n\in\N$, the set of CPDNN quantum channels $\Phi : M_n\to M_2$ coincides with the set of CPCP quantum channels $\Phi : M_n\to M_2$. Consequently, the normal form of Corollary~\ref{cor:channelnormal} classifies both classes of qubit-output quantum channels.
\end{corollary}

\section{Discussion and outlook}\label{sec:discussion}

Although the doubly nonnegative cone is strictly larger than the completely positive cone in higher matrix dimensions, this gap disappears for trace-preserving maps with qubit output. \Cref{thm:channelexact} resolves the qutrit-to-qubit question of Johnston and Sikora, but the proof shows more than the absence of a single counterexample. In the nontrivial binary-output regime, every CPDNN channel admits a complete normal form, and the free-channel cone becomes explicitly describable in terms of a vector $p\in[0,1]^n$ and a nonnegative matrix $B$ subject to one positive-semidefinite block constraint.

The mechanism is qubit-specific. For $i\neq j$, trace preservation forces $\Tr(J_{ij})=0$. When $J_{ij}$ is a $2\times 2$ entrywise nonnegative block, this implies that its diagonal entries vanish, leaving only one off-diagonal slot. That is exactly the collapse isolated in Lemma~\ref{lem:collapse}, and it is the reason why the entire channel can be encoded by a single matrix $B$. In output dimension $3$ or higher, trace-zero entrywise nonnegative blocks need not collapse in this way \cite{johnston2022completely}.
Thus the binary-output rigidity exploited here is the structural reason exactness holds.

The inequalities of Corollary~\ref{cor:tradeoff} provide an additional qubit-specific consequence. They show that the unique output coherence mode is quantitatively controlled by the diagonal splitting parameters. Equivalently, the normal form bounds the extent to which different input entries can contribute to the single output coherence slot.

Several questions remain open. The most immediate one is whether any analogue of this exactness phenomenon survives for output dimension $3$ or larger, perhaps under additional symmetry, sparsity, or normalization assumptions. More broadly, it would be interesting to determine which combinations of input dimension, output dimension, and physical constraints make the CPDNN relaxation exact for quantum channels. From the perspective of the resource theory of nonnegative amplitudes, these questions are part of the larger problem of understanding when tractable relaxations preserve the exact geometry and operational meaning of the free operations \cite{bomze2000copositive, brandao2015reversible}.



\bibliographystyle{IEEEtran}
\bibliography{references}

\end{document}